\documentclass{article}

\usepackage{amssymb}
\usepackage{amsmath}
\usepackage{amsfonts}
\usepackage{graphicx}
\usepackage{authblk}
\usepackage{ntheorem}

\theoremseparator{.}

\newtheorem{proposition}{Proposition} 
\newtheorem{lemma}{Lemma} 
\newtheorem{theorem}{Theorem}
\newtheorem{corollary}{Corollary}
\newtheorem{definition}{Definition}
\newtheorem{example}{Example}
\newtheorem{claim}{Claim}
\newtheorem{case}{Case}
\newtheorem{subcase}{Case}
\numberwithin{subcase}{case}
\newtheorem{remark}{Remark} 

\newenvironment{proof}[1][Proof]{\noindent\textbf{#1.} }{\ \rule{0.5em}{0.5em}}

\let\olddefinition\definition
\renewcommand{\definition}{\olddefinition\normalfont}
\let\oldexample\example
\renewcommand{\example}{\oldexample\normalfont}
\let\oldcase\case
\renewcommand{\case}{\oldcase\normalfont}
\let\oldsubcase\subcase
\renewcommand{\subcase}{\oldsubcase\normalfont}
\let\oldclaim\claim
\renewcommand{\claim}{\oldclaim\normalfont}
\let\oldremark\remark
\renewcommand{\remark}{\oldremark\normalfont}

\begin{document}

\begin{center}
{\Large Domination games played on line graphs of complete multipartite graphs}

\emph{H.G. Tananyan}

{\small Russian-Armenian (Slavonic) University, Yerevan, Armenian}

{\small E-mail: HTananyan@yahoo.com}
\end{center}
\sloppy

\begin{abstract}
The domination game on a graph $G$\ (introduced
by B. Bre\v{s}ar, S. Klav\v{z}ar, D.F. Rall \cite{BKR2010}) consists of two
players, Dominator and Staller, who take turns choosing a vertex from $G$\
such that whenever a vertex is chosen by either player, at least one
additional vertex is dominated. Dominator wishes to dominate the graph in as
few steps as possible, and Staller wishes to delay this process as much as
possible. The game domination number $\gamma _{{\small g}}(G)$\ is the
number of vertices chosen when Dominator starts the game; when Staller
starts, it is denoted by $\gamma _{{\small g}}^{\prime }(G).$

In this paper, the domination game on line graph $L\left( K_{\overline{%
m}}\right) $\ of complete multipartite graph $K_{\overline{m}}$\ $(%
\overline{m}\equiv (m_{1},...,m_{n})\in 
\mathbb{N}
^{n})$\ is considered, the exact values for game domination numbers are
obtained and optimal strategy for both players is described. Particularly,
it is proved that for $m_{1}\leq m_{2}\leq ...\leq m_{n}$ both $\gamma _{%
{\small g}}\left( L\left( K_{\overline{m}}\right) \right) =\min
\left\{ \left\lceil \frac{2}{3}\left\vert V\left( K_{\overline{m}%
}\right) \right\vert \right\rceil ,\right.$\ $\left. 2\max \left\{
\left\lceil \frac{1}{2}\left( m_{1}+...+m_{n-1}\right) \right\rceil ,\text{ }%
m_{n-1}\right\} \right\} -1$ when $n\geq 2$ and $\gamma _{g}^{\prime
}(L\left( K_{\overline{m}}\right) )=\min \left\{ \left\lceil \frac{2}{3%
}\left( \left\vert V(K_{_{\overline{m}}})\right\vert -2\right)
\right\rceil ,\right.$\ $\left. 2\max \left\{ \left\lceil \frac{1}{2}\left(
m_{1}+...+m_{n-1}-1\right) \right\rceil ,\text{ }m_{n-1}\right\} \right\} $
when $n\geq 4$.


\end{abstract}

\section{Introduction}

We consider only finite undirected\ graphs without loops and multi-edges.
The set of vertices of a graph $G$ is denoted by $V(G)$, and the set of
edges of $G,$ by $E(G)$. For a vertex $v\in V(G)$ the closed vertex
neighborhood is denoted by $N[v]=\{u\in V(G):(v,u)\in E(G)\}\cup \{v\}$ and
for an edge $e\in E(G),$ the closed edge neighborhood by $N[e]=\{e^{\prime
}\in E(G):e\neq ~e^{\prime },e~$and$~e^{\prime }~$are~adjacent in $G\}\cup
\{e\}$. The line graph of a graph $G$, denoted by $L(G)$, is the graph with
vertex set $E(G)$ in which two vertices are adjacent if and only if the
respective edges of $G$ have a vertex in common, i.e. $V(L(G))=E(G)$ and $%
E(L(G))=\{(e_{1},e_{2}):e_{1}\in E(G),e_{2}\in N[e_{1}],e_{1}\neq e_{2}\}$.
A complete graph on $m$ vertices is denoted by $K_m$, and a complete $n-$%
partite ($n\geq 2$) graph with partite classes $V_{1},V_{2},...V_{n}$ of
order $m_{1},m_{2},...,m_{n}$ respectively is denoted by $K_{_{\overline{m}%
}} $, where $\overline{m}=(m_{1},...,m_{n})$. Non-defined concepts can be
found in \cite{Harary1969}.

According to the terminology of \cite{BKR2010}-\cite{BKR2011}, we describe
two \textit{vertex domination games} and their edge-analogs played on a
finite graph $G$. In $Game$ $\mathfrak{D}_{v}$ two players, \textit{Dominator%
} and \textit{Staller}, alternate taking turns choosing a vertex from $G,$
with Dominator going first. Let $S$ denote the sequence of vertices $%
s_{1}s_{2}...$ chosen by the players. These vertices must be chosen in such
a way that whenever a vertex is chosen by either player, at least one
additional vertex of the graph $G$ is dominated that was not dominated by
the vertices previously chosen. That is, for each $i$:\qquad 
\begin{equation}
N[s_{i}]\backslash \overset{i-1}{\underset{j=1}{\bigcup }}N[s_{j}]\neq
\varnothing \qquad (1<i\leq |S|).  \label{eq:game_condition}
\end{equation}

In $Game~\mathfrak{D}_{v}^{\prime }$ the players alternate choosing vertices
satisfying to condition (\ref{eq:game_condition}) as in $Game$ $\mathfrak{D}%
_{v}$, except that Staller begins. Since the graph $G$ is finite, each of
the defined games will end in some finite number of moves regardless of how
the vertices are chosen. In each of the games, Dominator chooses vertices
using a strategy that will force the game to end in the fewest number of
moves, and Staller uses a strategy that will prolong the game as long as
possible. Following \cite{BKR2010}, we define the \textit{vertex game
domination number} of $G$, denoted by $\gamma _{g}(G)$, and the
Staller-start vertex game domination number of $G$, denoted by $\gamma
_{g}^{\prime }(G)$, to be the total number of vertices chosen when they play
respectively $Game$ $\mathfrak{D}_{v}$ and $Game$ $\mathfrak{D}_{v}^{\prime
} $ on graph $G$ using \textit{optimal strategies}.

In the Dominator-start \textit{edge domination game}, denoted by $Game$ $%
\mathfrak{D}_{e}$, and in the Staller-start edge domination game, denoted by 
$Game$ $\mathfrak{D}_{e}^{\prime }$, Dominator and Staller are taking edges
instead, under the condition (\ref{eq:game_condition}) where $%
S=s_{1}s_{2}...s_{|S|}$ is a sequence of chosen edges. Analogously, the 
\textit{edge game domination number} of $G$, denoted by $\gamma _{e,g}(G)$,
and the Staller-start edge game domination numbers of $G$, denoted by $%
\gamma _{g}^{\prime }(G)$, are the total numbers of edges chosen when they
play respectively $Game$ $\mathfrak{D}_{e}$ and $Game$ $\mathfrak{D}%
_{e}^{\prime }$ on graph $G$ using optimal strategies.

\begin{remark}
\label{rem:0.1} From definitions it immediately follows that $\gamma
_{g}(L(G))=\gamma _{e,g}(G)$ and $\gamma _{g}^{\prime }(L(G))=\gamma
_{e,g}^{\prime }(G)$ for every graph $G$.
\end{remark}

A set of covered vertices, denoted by $C_{S,i}$, at step $i$ ($1\leq i\leq
|S|$) in an instance $S=s_{1}s_{2}...s_{|S|}$ of $Game$ $\mathfrak{D}_{e}$
played on a graph $G$ is defined as a union of endpoints of chosen edges $%
s_{1},s_{2},...,s_{i}$. A vertex $v\in V(G)$ is called uncovered in $S$ at
step $i$ ($1\leq i\leq |S|$) if $v\not\in C_{S,i}$. For short, put $%
C_{S}=C_{S,|S|}$ and $C_{S,0}=\varnothing $.

In Section \ref{sec:Preliminaries}, helper properties for edge domination
games are given. In Section \ref{sec:K_m}, the game domination number when
at the end of the game at most one uncovered vertex remains is obtained and
as a corollary exact value of $\gamma _{g}(L(K_{m}))$ is calculated. In
Section \ref{sec:K_m_vector}, an semi-greedy strategy for Staller for edge
domination game played on complete multipartite graph is introduced. Through
that strategy, the lower bound for domination number, when at the end of the
game at least two uncovered vertices is left, is determined. Then from the
equality of the obtained upper and lower bounds, by using result from
Section \ref{sec:K_m}, game domination number $\gamma _{g}(L(K_{_{\overline{m%
}}}))$ is obtained, and the optimality of semi-greedy strategy for Staller
is shown. In Section \ref{sec:K_m_vector_Staller-start}, Staller-start game
domination number $\gamma _{g}^{\prime }(L(K_{_{\overline{m}}}))$ is
determined.

\section{Preliminaries and Basic Properties}

\label{sec:Preliminaries}

Following \cite{KWZ2011}, we use the following definitions. Let $G$ be a
graph on which several turns of the edge domination game have already been
taken. We say that a edge $e$ of $G$ is \textit{dominated} if some edge
within $N[e] $ has been played. A \textit{partially} edge dominated graph $%
G_{A}$ is a graph $G$ in which we suppose that some edges $A\subseteq E(G)$
have already been dominated, i.e. some moves have already been made,
although we are concerned with which edges have thus far been dominated,
rather than which have been chosen. If $G_{A}$ is a partially edge dominated
graph, then let $\gamma _{e,g}(G_{A})$ denote the number of turns remaining
in the game if Dominator has the next move. Similarly, let $\gamma
_{e,g}^{\prime }(G_{A})$ denote the number of turns remaining if Staller has
the next move.

On the basis of Remark \ref{rem:0.1}, the Continuation Principle (see \cite%
{KWZ2011}, Lemma 2.1) can be verbatim rewritten for partially edge dominated
graphs.

\begin{proposition}[Continuation Principle]
\label{rem:cont_princ} Let $G$ be a graph and $A\subseteq B\subseteq E(G).$
If $G_{A}$ and $G_{B}$ are the partially edge dominated graphs corresponding
to $G,$ with $A$ dominated and with $B$ dominated respectively, then $\gamma
_{e,g}(G_{A})\geq \gamma _{e,g}(G_{B})$ and $\gamma _{e,g}^{\prime
}(G_{A})\geq \gamma _{e,g}^{\prime }(G_{B}).$
\end{proposition}

\begin{proposition}
\label{prop:indep} Let $S$ be an instance of $Game$ $\mathfrak{D}_{e}$
played on a graph $G$. Then the vertices of
the set $V(G)\backslash C_{S,i}$ $(1 \leq i \leq |S|)$ are independent in $G$ if 
and only if game $S$ is over, i.e. $i = |S|$.
\end{proposition}

\begin{proof}
If $V(G)\backslash C_{S,i}$ ($1 \leq i \leq |S|$) is independent in $G$ then
all edges of $G$ are dominated and game $S$ is over, i.e. $i = |S|$.

If $v_{1},v_{2}\in V(G)\backslash C_{S,i}$ ($1 \leq i \leq |S|$) and 
$(v_{1},v_{2})\in E(G)$ then, since at step $i$ there
are no chosen edges adjacent to either $v_{1}$ or $v_{2}$, edge
$(v_{1},v_{2})$ is not dominated at step $i$. So, $i < |S|$.
Thus, $V(G)\backslash C_{S,|S|}$ is independent in $G$.
\end{proof}

\begin{proposition}
\label{prop:D_2} For every graph $G$ there exists an optimal strategy $%
\mathfrak{S}$ for $Game$ $\mathfrak{D}_{e}$ played on $G$ such that at each
step Dominator chooses an edge which covers exactly two new vertices, i.e.
for an arbitrary instance $S$ of $Game$ $\mathfrak{D}_{e}$ played on $G$
with strategy $\mathfrak{S}$ and for each odd $i$ ($1 \leq i \leq |S|$), $%
|C_{S,i}\backslash C_{S,i-1}|=2$.
\end{proposition}

\begin{proof}
Let at step $i$ ($1<i\leq \gamma _{e,g}(G))$ edges $E_{i}\subset E(G)$ are
dominated and Dominator by playing with an optimal strategy on move $i$
chooses edge $s_{i}$ which (by definition) dominating at least one new edge $%
s_{i}^{\prime }$. If edge $s_{i}$ covers two new vertices then in strategy $%
\mathfrak{S}$ Dominator will also choose $\ s_{i}$, otherwise Dominator will
choose edge $s_{i}^{\prime }$ instead of edge $s_{i}$, and since in that
case $E_{i}\cup N[s_{i}]\subseteq E_{i}\cup N[s_{i}^{\prime }]$, due to the\
Continuation Principle (see Propositions \ref{rem:cont_princ}), $\mathfrak{S}
$ is also optimal strategy.
\end{proof}

Let $dist(v,u)$ be the distance between vertices $v,u\in V(G)$. The
vertex-edge diameter of a connected graph $G$ (with $E\neq \varnothing $)
denoted by $diam(G)$ is defined as:%
\begin{equation}
diam(G)\equiv \underset{w\in V(G)}{\underset{(v,u)\in E(G)}{\max }}\min
\{dist(w,v),dist(w,u)\}.  \label{eq:vertex-edge-diameter}
\end{equation}

A strategy $\mathfrak{S}$ for $Game$ $\mathfrak{D}_{e}$ is called a \textit{%
2-1-}strategy if on each move, Dominator covers exactly two new vertices and
Staller covers exactly one, i.e. for an arbitrary instance $S$ of $Game$ $%
\mathfrak{D}_{e}$ played on $G$ with strategy $\mathfrak{S}$ and for each $i$
($1 \leq i \leq |S|$) both $|C_{S,i}\backslash C_{S,i-1}|=2$ when $i$ is odd
and $|C_{S,i}\backslash C_{S,i-1}|=1$ when $i$ is even.

\begin{proposition}
\label{prop:S_1} For every connected graph $G$ if $diam(G)=1$ then there
exists an optimal \textit{2-1-}strategy $\mathfrak{S}$ for $Game$ $\mathfrak{%
D}_{e}$ played on $G.$
\end{proposition}

\begin{proof}
Let Dominator plays with strategy $\mathfrak{S}$ as dercrabed in proof of
Propositions \ref{prop:D_2}. Let $E_{i}\subset E(G)$ be dominated edges at
step $i$ ($1<i\leq \gamma _{e,g}(G))$, let $v$ be a previously covered
vertex and let Staller by playing with an optimal strategy on move $i$
chooses edge $s_{i}=(u,w)$. Since $diam(G)=1$ either $(v,u)\in E(G)$ or $%
(v,w)\in E(G)$. If $s_{i}$ covers one new vertex then in strategy $\mathfrak{%
S}$ Staller will also choose $s_{i}$, otherwise Staller will choose edge $%
s_{i}^{\prime }$ (either $s_{i}^{\prime }=(v,u)$ if $(v,u)$ $\in E(G)$ or $%
s_{i}^{\prime }=(v,w)$ if $(v,w)$ $\in E(G)$) instead of edge $s_{i}$, and
since $E_{i}\cup N[s_{i}^{\prime }]\subseteq $ $E_{i}\cup N[s_{i}]$, due\ to
Continuation Principle (see Proposition \ref{rem:cont_princ}), $\mathfrak{S}$
is also optimal strategy.
\end{proof}

\begin{proposition}
\label{prop:even} Let $\gamma _{e,g}(G)$ be even. If there is an instance $S$
of $Game$ $\mathfrak{D}_{e}$ played on a graph $G$ with a \textit{2-1-}%
strategy then at the end of the game the number of uncovered vertices $%
V(G)/C_{S}$ is not less than $1$, i.e. $|V(G)|-|C_{S}|\geq 1$.
\end{proposition}

\begin{proof}
Since $S$ played with a \textit{2-1-}strategy and the last move was made by
Staller (because of $\gamma _{e,g}(G)$ is even), then on the last move
exactly one new vertex is covered, i.e. $|C_{S,\gamma
_{e,g}(G)}|=|C_{S,\gamma _{e,g}(G)-1}|+1$. Since $Game$ $\mathfrak{D}_{e}$ \
is not over at step $\gamma _{e,g}(G)-1,$ then due to Proposition \ref%
{prop:indep} $|C_{S,\gamma _{e,g}(G)-1}|\leq |V(G)|-2.$ Thus, $|C_{S,\gamma
_{e,g}(G)}|\leq |V(G)|-1.$
\end{proof}

\begin{proposition}
\label{prop:S_ineq} Let $S$ be an instance of $Game$ $\mathfrak{D}_{e}$
played on a graph $G$ with an optimal \textit{2-1-}strategy and let $%
S^{\prime }$ be an instance played on $G$ with a \textit{2-1-}strategy such
that Dominator plays optimally. Then 
\begin{equation}
|C_{S^{\prime }}|\leq |C_{S}|.  \label{eq:S_ineq}
\end{equation}
\end{proposition}

\begin{proof}
Since Dominator plays optimally in games $S$ and $S^{\prime },$ and Staller
plays optimally in game $S$, it immediately follows that $|S^{\prime }|\leq
|S|$. Since both $S$ and $S^{\prime }$ are played with \textit{2-1-}%
strategies, then (\ref{eq:S_ineq}).
\end{proof}

\begin{proposition}
\label{prop:S_is_optimal} Let $S$ and $S^{\prime }$ be instances of $Game$ $%
\mathfrak{D}_{e}$ played on graph $G$ with \textit{2-1-}strategies. If $%
\left\vert V(G)\backslash C_{S}\right\vert \leq 1$ and $\left\vert
V(G)\backslash C_{S^{\prime }}\right\vert \leq 1$ then $|S|=|S^{\prime }|$.
\end{proposition}

\begin{proof}
Let $|S|\neq |S^{\prime }|$. Since $\left\vert V(G)\backslash
C_{S}\right\vert \leq 1$ and $\left\vert V(G)\backslash C_{S^{\prime
}}\right\vert \leq 1$ then $\left\vert \left\vert C_{S}\right\vert
-\left\vert C_{S^{\prime }}\right\vert \right\vert $ $\leq 1$. On the other
hand, by Proposition \ref{prop:even}, $max\left\{ |S|,|S^{\prime }|\right\} $
is odd, as $min\left\{ \left\vert V(G)\backslash C_{S}\right\vert
,\left\vert V(G)\backslash C_{S^{\prime }}\right\vert \right\} =0$. So, $%
||C_{S}|-|C_{S^{\prime }}||\geq 2$ as $S$ and $S^{\prime }$ are \textit{2-1-}%
strategies. Hence, the obtained contradiction proves the proposition.
\end{proof}

\begin{proposition}
\label{prop:upper_bound} If $U\subset V(G)$ is an independent set in a
connected graph $G$ and $M\subset E(G)$ is a matching in induced subgraph $%
G[V(G)\backslash U]$, then%
\begin{equation}
\gamma _{e,g}\left( G\right) \leq 2(|V(G)\backslash U|-|M|)-1.
\label{eq:upper_bound_all}
\end{equation}
\end{proposition}

\begin{proof}
Since Dominator at most with $|M|+|V(G)\backslash U|-2|M|$ steps dominates
all edges of graph $G$, then upper bound (\ref{eq:upper_bound_all}) holds
immediately.
\end{proof}

\section{Domination game played on $L(K_{m})$}

\label{sec:K_m}

\begin{lemma}
\label{lem:main} If there is an instance $S$ of $Game$ $\mathfrak{D}_{e}$
played on a graph $G$ with an optimal \textit{2-1-}strategy such that $%
|V(G)|-|C_{S}|\leq 1$ then 
\begin{equation*}
\gamma _{e,g}(G)=\left\lceil \frac{2}{3}|V(G)|\right\rceil -1.
\end{equation*}
\end{lemma}

\setcounter{case}{0}

\begin{proof}
Consider the following three cases.

\begin{case}
$\gamma _{e,g}(G)$ is even.

From Proposition \ref{prop:even} it follows that $|C_{S}|=|V(G)|-1.$ Since $%
S $ played with an optimal \textit{2-1-}strategy and $\gamma _{e,g}(G)$ is
even, then $|C_{S}|=\frac{3}{2}\gamma _{e,g}(G).$ Hence, $\gamma _{e,g}(G)=%
\frac{2}{3}(|V(G)|-1)$. Accordingly, $|V(G)|=1\left( \text{mod }3\right) $.
\end{case}

\begin{case}
$\gamma _{e,g}(G)$ is odd and $|C_{S}|=|V(G)|-1$.

Since $\gamma _{e,g}(G)$ is odd, $|C_{S}|=\frac{3}{2}\left( \gamma
_{e,g}(G)-1\right) +2$. Hence, $\gamma _{e,g}(G)=\frac{2}{3}(|V(G)|-2)+1$.
So, $|V(G)|=2\left( \text{mod }3\right) $.
\end{case}

\begin{case}
$\gamma _{e,g}(G)$ is odd and $|C_{S}|=|V(G)|$.

Analogously, $\gamma _{e,g}(G)=\frac{2}{3}|V(G)|-1$. So, $|V(G)|=0\left( 
\text{mod }3\right) $.
\end{case}

Therefore, \textit{(a)} if $|V(G)|=0\left( \text{mod }3\right) $ then $%
\gamma _{e,g}(G)=\frac{2}{3}|V(G)|-1$; \textit{(b)} if $|V(G)|=1\left( \text{%
mod }3\right) $ then $\gamma _{e,g}(G)=\frac{2}{3}|V(G)|-\frac{2}{3}$; and 
\textit{(c)} if $|V(G)|=2\left( \text{mod }3\right) $ then $\gamma _{e,g}(G)=%
\frac{2}{3}|V(G)|-\frac{1}{3}$.
\end{proof}

\begin{theorem}
\label{cor:K_n} Let $m\in 
\mathbb{N}
.$ Then 
\begin{equation}
\gamma _{e,g}(K_{m})=\left\lceil \frac{2}{3}|V(K_{m})|\right\rceil -1.
\label{eq:cor_K_n}
\end{equation}
\end{theorem}

\begin{proof}
Let $m\geq 3$. Since $diam(K_{m})=1$, from Proposition \ref{prop:S_1} it
follows that there is an instance $S$ for $Game$ $\mathfrak{D}_{e}$ played
on $K_{m}$ with an optimal \textit{2-1-}strategy. Then, due to Proposition %
\ref{prop:indep}, $|C_{S}|\geq |V(K_{m})|-1$. Hence, formula (\ref%
{eq:cor_K_n}) immediately follows from Lemma \ref{lem:main}.
\end{proof}

\section{Domination game played on $L(K_{\overline{m}})$}

\label{sec:K_m_vector}

\begin{proposition}
\label{prop:diam_1} The vertex-edge diameter $diam(G)$, defined in formula (%
\ref{eq:vertex-edge-diameter}), of a connected graph $G$ with $|E(G)|\geq 2$
is equal to $1$ if and only if $G$ is a complete multipartite graph.
\end{proposition}

\begin{proof}
If $\overline{m}\in 
\mathbb{N}
^{n}$ and $\left\vert E(K_{_{\overline{m}}})\right\vert \geq 2$ then $%
diam\left( K_{_{\overline{m}}}\right) =1$. Thus, the sufficiency is proved.

Let $diam(G)=1$. A binary relationship $\alpha $ on $V(G)$ is defined as
follows: 
\begin{equation*}
v\alpha u\Leftrightarrow (v,u)\notin E(G)\qquad \forall v,u\in V(G).
\end{equation*}

It is trivial that $\alpha $ is reflexive ($v\alpha v$ for every $v\in V(G)$%
) and symmetric ($v\alpha u$ $\Rightarrow u\alpha v$ for every $v,u\in V(G)$%
). Assume $(v,u)\notin E(G)$ and $(u,w)\notin E(G)$ ($v,u,w\in V(G)$). Then, 
$(v,w)\notin E(G)$. Otherwise, if $(v,w)\in E(G),$ then due to $diam(G)=1,$
either $(v,u)\in E(G)$ or $(u,w)\in E(G).$ Hence $\alpha $ is transitive.
Thus, $\alpha $ is a relationship of equivalence. So, $V(G)$ can be
partitioned into disjoint sets $U_{1},...,U_{r},$ such that $U_{i}$ $(1\leq
i\leq r)$ is an independent set in $G$. Therefore, $G$ is isomorphic to $%
K_{\left( |U_{1}|,...,|U_{r}|\right) }$. Thus, the necessity is proved.
\end{proof}

\begin{definition}
Let $S$ be an instance of $Game$ $\mathfrak{D}_{e}$ played on graph $K_{_{%
\overline{m}}}$ ($\overline{m}\in 
\mathbb{N}
^{n}$, $n\geq 2$) and let for each $i$ ($1\leq i \leq |S|$) partite classes $%
V_{1},...,V_{n}$ of graph $K_{_{\overline{m}}}$ be renumbered as $%
V_{l_{1}^{\left( i\right) }},....V_{l_{n}^{\left( i\right) }}$ to satisfy
condition $|V_{l_{1}^{\left( i\right) }}\backslash C_{S,i-1}|\leq ...\leq
|V_{l_{n-1}^{\left( i\right) }}\backslash C_{S,i-1}|\leq |V_{l_{n}^{\left(
i\right) }}\backslash C_{S,i-1}|$. Then say that Staller plays $S$ with a 
\textit{semi-greedy strategy }if for each even $i$ ($1\leq i \leq |S|$)
Staller chooses an edge which covers exactly one new vertex $c_{i}$ which
satisfies to following conditions: 
\begin{equation*}
\begin{tabular}{ll}
$c_{i}\in V_{l_{n}^{\left( i\right) }}$ & when $|V_{l_{n}^{\left( i\right)
}}\backslash C_{S,i-1}|>|V_{l_{n-1}^{\left( i\right) }}\backslash
C_{S,i-1}|, $ \\ 
$c_{i}\in V(K_{_{\overline{m}}})\backslash \left( V_{l_{n-1}^{\left(
i\right) }}\cup V_{l_{n}^{\left( i\right) }}\right) $ & when $%
|V_{l_{n}^{\left( i\right) }}\backslash C_{S,i-1}|=|V_{l_{n-1}^{\left(
i\right) }}\backslash C_{S,i-1}|$ and $V(K_{_{\overline{m}}})\backslash
\left( V_{l_{n-1}^{\left( i\right) }}\cup V_{l_{n}^{\left( i\right)
}}\right) \backslash C_{S,i-1}\neq \varnothing ,$ \\ 
$c_{i}\in V_{l_{n-1}^{\left( i\right) }}\cup V_{l_{n}^{\left( i\right) }}$ & 
when $|V_{l_{n}^{\left( i\right) }}\backslash C_{S,i-1}|=|V_{l_{n-1}^{\left(
i\right) }}\backslash C_{S,i-1}|$ and $V(K_{_{\overline{m}}})\backslash
\left( V_{l_{n-1}^{\left( i\right) }}\cup V_{l_{n}^{\left( i\right)
}}\right) \backslash C_{S,i-1}=\varnothing .$%
\end{tabular}%
\end{equation*}
\end{definition}

\begin{proposition}[Lower Bound]
\label{prop:S_all} Let $n\geq 2,$ $\overline{m}\in 
\mathbb{N}
^{n},$ $m_{1}\leq m_{2}\leq ...\leq m_{n}$ and let $S$ be an instance of $%
Game$ $\mathfrak{D}_{e}$ played on graph $K_{_{\overline{m}}}$ with a 
\textit{2-1-}strategy such that (a) Dominator plays with an optimal strategy
and (b) Staller plays with a semi-greedy strategy. If at the end of the game
the number of uncovered vertices $V(K_{_{\overline{m}}})\backslash C_{S}$ is
not less than $2$ then 
\begin{equation}
\gamma _{e,g}\left( K_{_{\overline{m}}}\right) \geq |S|\geq 2\max \left\{
\left\lceil \frac{1}{2}\overset{n-1}{\underset{j=1}{\sum }}m_{j}\right\rceil
,m_{n-1}\right\} -1.  \label{eq:Lower_bound}
\end{equation}
\end{proposition}

\begin{proof}
Since Proposition \ref{prop:indep}, there is partite class $V_{l}$ such that 
$V(K_{_{\overline{m}}})\backslash C_{S}\subseteq V_{l}$.

\begin{claim}
\label{claim:S_odd} The number $|S|$ is odd.
\end{claim}

\begin{proof}
Let $|S|$ be even. Then at last step exactly one new vertex $w\in V(K_{_{%
\overline{m}}})\backslash C_{S,|S|-1}$ is covered, as $S$ is played with a
2-1-strategy. By Proposition \ref{prop:indep}, $w\notin V_{l}$, i.e. there
is an index $l^{\prime }$ such that $l^{\prime }\neq l$ and $w\in
V_{l^{\prime }}$. Since $|V_{l^{\prime }}\backslash C_{S,|S|-1}|=1$ and $%
|V_{l}\backslash C_{S,|S|-1}|\geq 2$, from Staller's strategy it follows
that in $S$ at last step must be chosen vertex from $V_{l}$ and game will
not be over at step $|S|$. Thus, the obtained contradiction proves Claim \ref%
{claim:S_odd}.
\end{proof}

\begin{claim}
\label{claim:S_all_from_l} For $i=1,3,...,|S|$%
\begin{equation}
\left\vert V_{l}\backslash C_{S,i}\right\vert >\left\vert V_{k}\backslash
C_{S,i}\right\vert \qquad k=1,2,...,n;k\neq l.  \label{eq:max_indep_set}
\end{equation}
\end{claim}

\begin{proof}
Since $V(K_{_{\overline{m}}})\backslash C_{S}\subseteq V_{l}$, (\ref%
{eq:max_indep_set}) holds when $|S|=1$. Hence, assume $|S|>1$. Claim \ref%
{claim:S_all_from_l} when $|S|>1$ will be proved by a contrary assumption.
It is assumed there exist some even $p$ ($1<p<|S|$) and partite class $%
V_{l^{\prime }}$ such that%
\begin{equation}
\left\vert V_{l}\backslash C_{S,i}\right\vert >\left\vert V_{k}\backslash
C_{S,i}\right\vert \qquad k=1,2,...,n;k\neq l;i=p+1,p+3,...,|S|,
\label{eq:max_indep_set_p}
\end{equation}%
and%
\begin{equation}
\left\vert V_{l}\backslash C_{S,p-1}\right\vert \leq \left\vert V_{l^{\prime
}}\backslash C_{S,p-1}\right\vert .
\end{equation}

From inequalities (\ref{eq:max_indep_set_p}), due to Staller's strategy,
follows that 
\begin{equation}
c_{i}\in V_{l}\qquad i=p+2,p+4,...,|S|-1.  \label{eq:S_all_*}
\end{equation}

Let $f$ be the number of remaining moves for Staller to complete the game
after $p^{th}$ move, i.e. $f\equiv \frac{1}{2}(|S|-1-p)$, and since $|S|$ is
odd, Dominator needs $f+1$ moves to complete the game. On the strength of
Staller's strategy, consider the following three cases.

\setcounter{case}{0}

\begin{case}
$\left\vert V_{l}\backslash C_{S,p-1}\right\vert <\left\vert V_{l^{\prime
}}\backslash C_{S,p-1}\right\vert $.

Since (a) at each step Dominator can cover at most one vertex from the
independent set $V_{l^{\prime }}$, (b) at step $p$ Staller can cover at most
one vertex from $V_{l^{\prime }}$ and (c) in remaining $f$ moves Staller
covers vertices only from $V_{l}$ (see (\ref{eq:S_all_*})), then 
\begin{equation}
\left\vert V_{l^{\prime }}\backslash C_{S}\right\vert \geq \left\vert
V_{l^{\prime }}\backslash C_{S,p-1}\right\vert -\left( f+1\right) -1.
\label{eq:S_all_1}
\end{equation}%
and 
\begin{equation}
\left\vert V_{l}\backslash C_{S}\right\vert \leq \left\vert V_{l}\backslash
C_{S,p-1}\right\vert -f.  \label{eq:S_all_2}
\end{equation}

Since (\ref{eq:S_all_1}) and (\ref{eq:S_all_2}), 
\begin{equation*}
\left\vert V_{l^{\prime }}\backslash C_{S}\right\vert \geq \left\vert
V_{l^{\prime }}\backslash C_{S,p-1}\right\vert -f-2\geq \left\vert
V_{l}\backslash C_{S,p-1}\right\vert +1-f-2\geq \left\vert V_{l}\backslash
C_{S}\right\vert -1\geq 1,
\end{equation*}%
which contradicts to $\left\vert V_{l^{\prime }}\backslash C_{S}\right\vert
=0$ (since $V(K_{_{\overline{m}}})\backslash C_{S}\subseteq V_{l}$). Thus,
case 1 is impossible.
\end{case}

\begin{case}
$\left\vert V_{l}\backslash C_{S,p-1}\right\vert =\left\vert V_{l^{\prime
}}\backslash C_{S,p-1}\right\vert $ and $c_{p}\notin V_{l^{\prime }}$.

From (\ref{eq:S_all_*}) it follows that (\ref{eq:S_all_2}) holds. Since at
each remaining step Dominator can cover at most one vertex from the
independent set $V_{l^{\prime }}$, if (\ref{eq:S_all_*}) is taken into
account, then 
\begin{equation}
\left\vert V_{l^{\prime }}\backslash C_{S}\right\vert \geq \left\vert
V_{l^{\prime }}\backslash C_{S,p-1}\right\vert -\left( f+1\right) .
\label{eq:S_all_4}
\end{equation}

From inequalities (\ref{eq:S_all_2}) and (\ref{eq:S_all_4}) it follows that 
\begin{equation*}
0=\left\vert V_{l^{\prime }}\backslash C_{S}\right\vert \geq \left\vert
V_{l^{\prime }}\backslash C_{S,p-1}\right\vert -f-1\geq \left\vert
V_{l}\backslash C_{S,p-1}\right\vert -f-1\geq \left\vert V_{l}\backslash
C_{S}\right\vert -1\geq 1,
\end{equation*}%
which is contradictory. Thus, case 2 is also impossible.
\end{case}

\begin{case}
$\left\vert V_{l}\backslash C_{S,p-1}\right\vert =\left\vert V_{l^{\prime
}}\backslash C_{S,p-1}\right\vert $ and $c_{p}\in V_{l^{\prime }}$.

Since Staller's strategy $V(K_{_{\overline{m}}})\backslash \left( V_{l}\cup
V_{l^{\prime }}\right) \backslash C_{S,p-1}=\varnothing \,$, so at each step
Dominator covers one vertex from both independent set $V_{l}$ and
independent set $V_{l^{\prime }}$, if (\ref{eq:S_all_*}) is taken into
account, then 
\begin{equation}
\left\vert V_{l}\backslash C_{S}\right\vert =\left\vert V_{l}\backslash
C_{S,p-1}\right\vert -f-\left( f+1\right) .  \label{eq:S_all_3}
\end{equation}%
and (\ref{eq:S_all_1}) holds. Inequalities (\ref{eq:S_all_1}) and (\ref%
{eq:S_all_3}) yield contradictory $0=\left\vert V_{l^{\prime }}\backslash
C_{S}\right\vert \geq \left\vert V_{l^{\prime }}\backslash
C_{S,p-1}\right\vert -f-2=\left\vert V_{l}\backslash C_{S,p-1}\right\vert
-f-2=\left\vert V_{l}\backslash C_{S}\right\vert +f-1=1+f\geq 1.$ Thus, case
3 is impossible as well.
\end{case}

Thus, the obtained contradictions prove Claim \ref{claim:S_all_from_l}.
\end{proof}

\begin{claim}
\label{claim:S_V_n=V_l} $\left\vert V_{l}\right\vert =\left\vert
V_{n}\right\vert $.
\end{claim}

\begin{proof}
Since Claim \ref{claim:S_all_from_l}, either $n=l$ or $\left\vert
V_{l}\backslash C_{S,1}\right\vert >\left\vert V_{n}\backslash
C_{S,1}\right\vert $. So, $|V_{l}|=|V_{n}|$ as $\left\vert V_{n}\cap
C_{S,1}\right\vert \leq 1$.
\end{proof}

In virtue of Claim \ref{claim:S_V_n=V_l}, assume that $l=n$. So, from Claims %
\ref{claim:S_all_from_l} and \ref{claim:S_V_n=V_l} it follows that in $S$
Stellar does not cover vertices from $V\left( K_{_{\overline{m}}}\right)
\backslash V_{n}$. Hence, on the one hand, since Dominator needs at least $%
m_{n-1}$ steps to cover all vertices of independent set $V_{n-1}$ to
complete the game, $\gamma _{e,g}\left( K_{_{\overline{m}}}\right) \geq
|S|\geq 2m_{n-1}-1$. On the other hand, as Dominator covers exactly two new
vertices at each step, Dominator needs at least $\left\lceil \frac{1}{2}%
\left\vert V\left( K_{_{\overline{m}}}\right) \backslash V_{n}\right\vert
\right\rceil $ steps to cover all vertices of $V\left( K_{_{\overline{m}%
}}\right) \backslash V_{n}$. So, $\gamma _{e,g}\left( K_{_{\overline{m}%
}}\right) \geq |S|\geq 2\left\lceil \frac{1}{2}\left\vert V\left( K_{_{%
\overline{m}}}\right) \backslash V_{n}\right\vert \right\rceil -1$. Thus,
lower bound (\ref{eq:Lower_bound}) holds.
\end{proof}

\begin{example}
Let $G$ be a graph with vertaices $\{v_{1},...,v_{7}\}$ and edges $%
\{(v_{1},v_{3}),(v_{2},v_{3}),(v_{3},v_{4}),(v_{4},v_{5}),(v_{5},v_{6}),(v_{5},v_{7})\} 
$. Although, $S=(v_{3},v_{4})(v_{4},v_{5})$ is an instance of $Game$ $%
\mathfrak{D}_{e}$ played on $G$ with an optimal 2-1-strategy such that $%
|V(G)\backslash C_{S}|=4$ but $|S|$ is even. So, Claim \ref{claim:S_odd}
does not work for $S$ because of $dima(G)\neq 1$ and Staller could not play
with semi-greedy strategy.
\end{example}

\begin{proposition}[Upper Bound]
\label{prop:rest>2} Let $n\geq 2,$ $\overline{m}\in 
\mathbb{N}
^{n},$ $m_{1}\leq m_{2}\leq ...\leq m_{n}$. Then%
\begin{equation}
\gamma _{e,g}\left( K_{_{\overline{m}}}\right) \leq 2\max \left\{
\left\lceil \frac{1}{2}\overset{n-1}{\underset{j=1}{\sum }}m_{j}\right\rceil
,m_{n-1}\right\} -1.  \label{eq:value_in_case_>2}
\end{equation}
\end{proposition}

\begin{proof}
Put $\sigma _{0}\equiv 0$ and $\sigma _{k}\equiv m_{1}+...+m_{k}$ for $%
k=1,...,n$. On the one hand, if $m_{n-1}\leq \sigma _{n-2}$ then in subgraph 
$K_{(m_{1},...,m_{n-1})}$ of $K_{_{\overline{m}}}$ there is a matching with $%
\left\lfloor \frac{1}{2}\sigma _{n-1}\right\rfloor $ edges (see \cite%
{SITTON1996}). Hence, from Proposition \ref{prop:upper_bound} it follows
that $\gamma _{e,g}\left( K_{_{\overline{m}}}\right) \leq 2\left( \left\vert
V\left( K_{_{\overline{m}}}\right) \backslash V_{n}\right\vert -\left\lfloor 
\frac{1}{2}\sigma _{n-1}\right\rfloor \right) -1=2\left\lceil \frac{1}{2}%
\sigma _{n-1}\right\rceil -1$ when $m_{n-1}\leq \sigma _{n-2}$. On the other
hand, if $m_{n-1}>\sigma _{n-2}$ then in subgraph $K_{(m_{1},...,m_{n-1})}$
of $K_{_{\overline{m}}}$ there is a matching with $\sigma _{n-2}$ edges.
Hence, from Proposition \ref{prop:upper_bound} it follows that $\gamma
_{e,g}\left( K_{_{\overline{m}}}\right) \leq 2\left( \left\vert V\left( K_{_{%
\overline{m}}}\right) \backslash V_{n}\right\vert -\sigma _{n-2}\right)
-1=2m_{n-1}-1$ when $m_{n-1}>\sigma _{n-2}$. Thus, $\gamma _{e,g}\left( K_{_{%
\overline{m}}}\right) \leq 2\max \left\{ \left\lceil \frac{1}{2}\sigma
_{n-1}\right\rceil ,m_{n-1}\right\} -1$ in both cases.
\end{proof}

\begin{theorem}
\label{th:main} Let $n\geq 2,\overline{m}\in 
\mathbb{N}
^{n}$ and $m_{1}\leq m_{2}\leq ...\leq m_{n}$. Then 
\begin{equation*}
\gamma _{e,g}(K_{_{\overline{m}}})=\min \left\{ \left\lceil \frac{2}{3}%
\left\vert V(K_{_{\overline{m}}})\right\vert \right\rceil ,\text{ }2\max
\left\{ \left\lceil \frac{1}{2}\overset{n-1}{\underset{j=1}{\sum }}%
m_{j}\right\rceil ,m_{n-1}\right\} \right\} -1.
\end{equation*}
\end{theorem}

\begin{proof}
Let $S$ be an instance of $Game$ $\mathfrak{D}_{e}$ played on graph $K_{_{%
\overline{m}}}$ with an optimal \textit{2-1-}strategy.

\setcounter{case}{0}

\begin{case}
$\left\vert V(K_{_{\overline{m}}})\backslash C_{S}\right\vert \geq 2$.

Since $diam\left( K_{_{\overline{m}}}\right) =1$, there exists an instance $%
S^{\prime }$ of $Game$ $\mathfrak{D}_{e}$ played on graph $K_{_{\overline{m}%
}}$ with a \textit{2-1-}strategy such that Dominator plays with an optimal
strategy and Staller plays with a semi-greedy strategy. From Proposition \ref%
{prop:S_ineq} it follows that $|V(K_{_{\overline{m}}})\backslash
C_{S^{\prime }}|\geq |V(K_{_{\overline{m}}})\backslash C_{S}|\geq 2$. Hence,
from Propositions \ref{prop:S_all} and \ref{prop:rest>2} follows that 
\begin{equation*}
\gamma _{e,g}\left( K_{_{\overline{m}}}\right) =2\max \left\{ \left\lceil 
\frac{1}{2}\overset{n-1}{\underset{j=1}{\sum }}m_{j}\right\rceil
,m_{n-1}\right\} -1.
\end{equation*}%
Since $|S^{\prime }|$ is odd (see Claim \ref{claim:S_odd} inside Proposition %
\ref{prop:S_all}), $\frac{3}{2}(|S^{\prime }|-1)+2=\left\vert C_{S^{\prime
}}\right\vert \leq \left\vert V(K_{_{\overline{m}%
}})\right\vert -2$. So, $\gamma _{e,g}\left( K_{_{\overline{m}}}\right)
=|S^{\prime }|\leq \left\lceil \frac{2}{3}\left\vert V(K_{_{\overline{m}%
}})\right\vert \right\rceil -1$.
\end{case}

\begin{case}
$\left\vert V(K_{_{\overline{m}}})\backslash C_{S}\right\vert \leq 1$.

From Lemma \ref{lem:main} follows that $\gamma _{e,g}\left( K_{_{\overline{m}%
}}\right) =\left\lceil \frac{2}{3}\left\vert V(K_{_{\overline{m}%
}})\right\vert \right\rceil -1$. Hence, from Proposition \ref{prop:S_all}
follows that 
\begin{equation*}
\left\lceil \frac{2}{3}\left\vert V(K_{_{\overline{m}}})\right\vert
\right\rceil -1=\gamma _{e,g}\left( K_{_{\overline{m}}}\right) \leq 2\max
\left\{ \left\lceil \frac{1}{2}\overset{n-1}{\underset{j=1}{\sum }}%
m_{j}\right\rceil ,m_{n-1}\right\} -1.
\end{equation*}
\end{case}

Thus, the proof is completed.
\end{proof}

\begin{corollary}
\label{prop:semi-greedy_is_optimal} Let $S$ be an instance of $Game$ $%
\mathfrak{D}_{e}$ played on graph $K_{_{\overline{m}}}$ ($n\geq 2$) with 
\textit{2-1-}strategy such that (a) Dominator plays with an optimal strategy
and (b) Staller plays with a semi-greedy strategy. Then $|S|=\gamma
_{e,g}\left( K_{_{\overline{m}}}\right) $.
\end{corollary}

\begin{proof}
From Propositions \ref{prop:S_all} and \ref{prop:rest>2} it follows that $%
|S|=\gamma _{e,g}\left( K_{_{\overline{m}}}\right) $ when $\left\vert V(K_{_{%
\overline{m}}})\backslash C_{S}\right\vert \geq 2$. Let $S^{\prime }$ be an
instance of $Game$ $\mathfrak{D}_{e}$ played on graph $K_{_{\overline{m}}}$
with an optimal \textit{2-1-}strategy. Hence, from Proposition \ref%
{prop:S_is_optimal} it follows that $|S|=|S^{\prime }|=\gamma _{e,g}\left(
K_{_{\overline{m}}}\right) $ when $\left\vert V(K_{_{\overline{m}%
}})\backslash C_{S}\right\vert \leq 1$, as $\left\vert V(K_{_{\overline{m}%
}})\backslash C_{S^{\prime }}\right\vert \leq \left\vert V(K_{_{\overline{m}%
}})\backslash C_{S}\right\vert $ due to Proposition \ref{prop:S_ineq}. Thus,
semi-greedy strategy is an optimal strategy for Staller for $Game$ $%
\mathfrak{D}_{e}$ played on graph $K_{_{\overline{m}}}$.
\end{proof}

\begin{example}
Since Theorem \ref{th:main}, $\gamma _{e,g}(K_{_{2,2,6,6}})=10$, $\gamma
_{e,g}(K_{_{2,2,4,5}})=7$ and $\gamma _{e,g}(K_{_{1,2,5,5}})=8$. Hence, just
"greedy" strategy for Staller, when at each step Staller choose edge to
cover vertex from some maximum independent set, is not optimal and it is
expedient to use semi-greedy strategy instead.
\end{example}

\begin{corollary}
\label{prop:greedy_is_optimal_for_Dominator} Let $n\geq 2$, $\overline{m}\in 
\mathbb{N}
^{n}$, $m_{1}\leq m_{2}\leq ...\leq m_{n}$, let $M$ be a maximal matching in
induced subgraph $K_{_{\overline{m}}}[V(K_{_{\overline{m}}})\backslash
V_{n}] $ and let $S=s_{1}...s_{|S|}$ be an instance of $Game$ $\mathfrak{D}%
_{e}$ played on graph $K_{_{\overline{m}}}$ with \textit{2-1-}strategy such
that Staller plays with an optimal strategy. If Dominator at each step $i$ ($%
i\leq |S|$) chooses edge from $M$ when $M\backslash \left\{
s_{1},...,s_{i-1}\right\} \neq \varnothing $ then $|S|=\gamma _{e,g}\left(
K_{_{\overline{m}}}\right) $.
\end{corollary}

\begin{proof}
Let $S^{\prime }$ be an instance of $Game$ $\mathfrak{D}_{e}$ played on
graph $K_{_{\overline{m}}}$ with an optimal \textit{2-1-}strategy. From
Propositions \ref{prop:S_all} and \ref{prop:rest>2} it follows that $%
|S^{\prime }|=|S|=\gamma _{e,g}\left( K_{_{\overline{m}}}\right) $ when $\left\vert V(K_{_{%
\overline{m}}})\backslash C_{S^{\prime }}\right\vert \geq 2$. On the other
hand, from Proposition \ref{prop:S_is_optimal} it follows that $%
|S|=|S^{\prime }|=\gamma _{e,g}\left( K_{_{\overline{m}}}\right) $ when $%
\left\vert V(K_{_{\overline{m}}})\backslash C_{S^{\prime }}\right\vert \leq
1 $, as $\left\vert V(K_{_{\overline{m}}})\backslash C_{S}\right\vert \leq
\left\vert V(K_{_{\overline{m}}})\backslash C_{S^{\prime }}\right\vert $ due
to Proposition \ref{prop:S_ineq}. 
\end{proof}

\section{Staller-start domination game played on $L(K_{\overline{m}})$}

\label{sec:K_m_vector_Staller-start}

\begin{proposition}
\label{prop:Staller_start} Let $n\geq 3$, $\overline{m}\in 
\mathbb{N}
^{n}$ and $m_{1}\leq m_{2}\leq ...\leq m_{n}$. Then%
\begin{equation}
\gamma _{e,g}^{\prime }(K_{_{\overline{m}}})=\min \left\{ \left\lceil \frac{2%
}{3}\left( \left\vert V(K_{_{\overline{m}}})\right\vert -2\right)
\right\rceil ,\text{ }2\max \left\{ \left\lceil \frac{1}{2}\left( \overset{%
n-1}{\underset{j=1}{\sum }}m_{j}-1-\eta \right) \right\rceil ,\text{ }%
m_{n-1}-\mu \right\} \right\} ,  \label{eq:Staller_started}
\end{equation}%
where $\mu $ equals $1$ when $n=3$ and $m_{n-1}=m_{n}$, and otherwise $\mu $
equals $0$; and $\eta $ equals $1$ when $m_{n-1}=m_{n}$, and otherwise $\eta 
$ equals $0$.
\end{proposition}

\begin{proof}
For $r=1,...,n$ put $\overline{m}^{\left( r\right) }$ $\equiv \left(
m_{1},...,m_{r-1},m_{r}-1,m_{r+1},...m_{n}\right) $. Since $\gamma
_{e,g}^{\prime }(K_{_{\overline{m}}})=\underset{1\leq r<t\leq n}{\max }%
\left\{ \gamma _{e,g}(K_{_{\overline{m}^{\left( r\right) \left( t\right)
}}})\right\} +1$, from Theorem \ref{th:main} and from the equality%
\begin{equation*}
\underset{1\leq z\leq q}{\max }\left\{ \min \left\{ a,\max
\{b_{z},c_{z}\right\} \}\right\} =\min \left\{ a,\max \left\{ \underset{%
1\leq z\leq q}{\max }\{b_{z}\},\underset{1\leq z\leq q}{\max }%
\{c\,_{z}\}\right\} \right\} \qquad \forall q\in 
\mathbb{N}
;a,b_{1},...,b_{q},c_{1},...,c_{q}\in 
\mathbb{R}
;
\end{equation*}%
it follows that (\ref{eq:Staller_started}) holds.
\end{proof}

\begin{theorem}
\label{th:main_S_start} Let $n\geq 4$, $\overline{m}\in 
\mathbb{N}
^{n}$ and $m_{1}\leq m_{2}\leq ...\leq m_{n}$. Then%
\begin{equation}
\gamma _{e,g}^{\prime }(K_{_{\overline{m}}})=\min \left\{ \left\lceil \frac{2%
}{3}\left( \left\vert V(K_{_{\overline{m}}})\right\vert -2\right)
\right\rceil ,\text{ }2\max \left\{ \left\lceil \frac{1}{2}\left( \overset{%
n-1}{\underset{j=1}{\sum }}m_{j}-1\right) \right\rceil ,\text{ }%
m_{n-1}\right\} \right\} .  \label{eq:main_S_start}
\end{equation}
\end{theorem}

\begin{proof}
From Proposition \ref{prop:Staller_start} it immediately follows that (\ref%
{eq:main_S_start}) holds when $m_{n-1}\neq m_{n}$. Put $\sigma _{k}\equiv
m_{1}+...+m_{k}$ for $k=1,...,n$. From $m_{n-1}=m_{n}$ it follows that (a)
if $\frac{1}{2}\left( \sigma _{n-1}-2\right) \leq m_{n-1}$ then $\frac{2}{3}%
\left( \sigma _{n-1}+m_{n}-2\right) \leq \frac{2}{3}\left(
2m_{n-1}+2+m_{n}-2\right) =2m_{n-1}$ and (b) if $\frac{1}{2}\left( \sigma
_{n-1}-2\right) >m_{n-1}$ then $\frac{2}{3}\left( \sigma
_{n-1}+m_{n}-2\right) <\frac{2}{3}\left( \sigma _{n-1}+\frac{1}{2}\left(
\sigma _{n-1}-2\right) -2\right) \leq 2\left\lceil \frac{1}{2}\left( \sigma
_{n-1}-2\right) \right\rceil $. Thus, formulas (\ref{eq:Staller_started})
and (\ref{eq:main_S_start}) are equivalent when $m_{n-1}=m_{n}$.
\end{proof}

\begin{remark}
\label{rem:main_n=2} If $m_{2}\geq 1$ then $\gamma _{e,g}^{\prime
}(K_{1,m_{2}})=1$ and if $m_{2}\geq m_{1}\geq 2$ then $\gamma _{e,g}^{\prime
}(K_{m_{1},m_{2}})=\gamma _{e,g}(K_{m_{1}-1,m_{2}-1})+1$.
\end{remark}

\begin{remark}
\label{rem:n=3} Since Proposition \ref{prop:Staller_start}, $\gamma
_{e,g}^{\prime }(K_{m_{1},m_{2},m_{3}})=\min \left\{ \left\lceil \frac{2}{3}%
\left( \left\vert V(K_{m_{1},m_{2},m_{3}})\right\vert -2\right) \right\rceil
,\text{ }2(m_{2}-1)\right\} $ when $m_{1}\leq m_{2}=m_{3}$, and $\gamma
_{e,g}^{\prime }(K_{m_{1},m_{2},m_{3}})=\min \left\{ \left\lceil \frac{2}{3}%
\left( \left\vert V(K_{m_{1},m_{2},m_{3}})\right\vert -2\right) \right\rceil
,\text{ }2m_{2}\right\} $ when $m_{1}\leq m_{2}<m_{3}$.
\end{remark}

\section*{Acknowledgments}

The author would like to thank his colleague Hrant Khachatryan for
introducing him to the domination game sphere.

\end{document}